\documentclass[final,1p,times,authoryear]{elsarticle}
\pdfoutput=1

\usepackage{comment}

\makeatletter
\def\ps@pprintTitle{\let\@oddhead\@empty
     \let\@evenhead\@empty
     \def\@oddfoot{}
     \let\@evenfoot\@oddfoot}
\makeatother

\usepackage{hyperref}
\usepackage{natbib}
\usepackage[svgnames]{xcolor}
\usepackage{amsmath}
\usepackage{amsthm}
\usepackage{amssymb}
\usepackage{pdfpages}
\usepackage{enumitem}
\usepackage{xspace}

\usepackage{hyphenat} 

\usepackage{algorithm}
\usepackage{algpseudocode}
\usepackage{algorithmicx}
\usepackage{dutchcal}
\usepackage{bm}

\newcommand{\Fstar}{ \pmat{F}^{*}}

 \DeclareMathOperator{\Syl}{ \mbox{\upshape {Syl}}}

\newcommand{\softO}{{O\tilde{\phantom{\imath}}}}

\newcommand{\mat}[1]{{#1}} \newcommand{\pmat}[1]{{\mathcal{#1}}}

\newcommand{\companion}{{\phi}}

\newcommand{\conv}{\companion}

\DeclareMathOperator{\Adj}{Adj}

 \DeclareMathOperator{\lcoeff}{{lcoeff}}
\newcommand{\A}{{\pmat{A}}} \renewcommand{\P}{{\pmat{P}}}

\newcommand{\B}{{\pmat{B}}}

\newcommand{\bb}{\pmat{b}}

\newcommand{\Span}{\mbox{\upshape{span}}}
\DeclareMathOperator{\SNF}{SNF} 

\newcommand{\struct}{{\bm{\Delta}}}

\DeclareMathOperator{\rank}{rank} \DeclareMathOperator{\diag}{diag}

\newcommand{\Otil}{{\widetilde{O}}} 

\renewcommand{\vec}{{\mbox{\upshape vec}}}

\newcommand{\pvec}{{\mbox{\upshape pvec}}}

\newcommand{\dist}{\mbox{\upshape{dist}}}

\newcommand{\divs}{{\mskip3mu|\mskip3mu}}

\newcommand{\dd}{{\bf {d}}}

\newcommand{\R}{{\mathsf{R}}}
\newcommand{\Tr}{{\mbox{\upshape Tr}}}
\newcommand{\RR}{{\mathbb{R}}}
\newcommand{\CC}{{\mathbb{C}}}
\newcommand{\QQ}{{\mathbb{Q}}}

\newcommand{\D}{\partial}
\newcommand{\norm}[1]{{\|#1\|}}
\newcommand{\tallnorm}[1]{ {\left \|#1 \right \|}}
\newcommand{\nxn}{{n\times n}}
\newcommand{\ftil}{\widetilde{f}}
\newcommand{\gtil}{\widetilde{g}}
\newcommand{\f}{{\bf {f}}}
\newcommand{\rev}{\mbox{\upshape{rev}}}

\RequirePackage[disable,textsize=small,textwidth=1.5\marginparwidth]{todonotes}

\makeatletter
\renewcommand{\todo}[2][]{\@bsphack\@todo[{#1}]{#2}\@esphack}\makeatother

\newcommand\scalemath[2]{\scalebox{#1}{\mbox{\ensuremath{\displaystyle #2}}}}

\newtheorem{theorem}{Theorem}[section]
\newtheorem{corollary}[theorem]{Corollary}
\newtheorem{lemma}[theorem]{Lemma}

\newtheorem{definition}[theorem]{Definition}
 \newtheorem{example}[theorem]{Example}
\newtheorem{problem}[theorem]{Problem} \newtheorem*{problem*}{Problem}

\newtheorem{conjecture*}{Conjecture} \newtheorem{remark}[theorem]{Remark}
\newtheorem*{remark*}{Remark}

\definecolor{darkgreen}{rgb}{0,.35,0} \definecolor{darkblue}{rgb}{0,0,.5}
\definecolor{darkred}{rgb}{.6,0,0}

\usepackage{hyperref} \hypersetup{pdfauthor={Giesbrecht, Haraldson, Labahn}
  pdftitle={Computing Nearby Non-trivial Smith Forms},
  bookmarksnumbered=true,colorlinks,linkcolor=darkblue,citecolor=darkgreen,
  urlcolor=darkred}

\begin{document}

\begin{frontmatter}

  \title{Computing Nearby Non-trivial Smith Forms}
  
  \author{Mark Giesbrecht}
\ead{mwg@uwaterloo.ca} \ead[url]{https://cs.uwaterloo.ca/~mwg}
  \ead{mwg@uwaterloo.ca}

  \author{Joseph Haraldson}
\ead[url]{https://cs.uwaterloo.ca/~jharalds} \ead{jharalds@uwaterloo.ca}

  \author{George Labahn} \address{Cheriton School of Computer Science,
    University of Waterloo, Waterloo, Canada}
\ead[url]{https://cs.uwaterloo.ca/~glabahn} \ead{glabahn@uwaterloo.ca}

\begin{abstract}
  We consider the problem of computing the nearest matrix polynomial with a
  non-trivial Smith Normal Form. We show that computing the Smith form of a
  matrix polynomial is amenable to numeric computation as an optimization
  problem. Furthermore, we describe an effective optimization technique to
  find a nearby matrix polynomial with a non-trivial Smith form.  The
  results are then generalized to include the computation of a matrix
  polynomial having a maximum specified number of ones in the Smith Form
  (i.e., with a maximum specified McCoy rank).

  We discuss the geometry and existence of solutions and how our results
  can be used for an error analysis.  We develop an
  optimization-based approach and demonstrate an iterative numerical method
  for computing a nearby matrix polynomial with the desired spectral
  properties.  We also describe an implementation of our algorithms and
  demonstrate the robustness with examples in \texttt{Maple}.
\end{abstract}

\end{frontmatter}

\section{Introduction}
\label{sec:intro}

For a given square matrix polynomial $\A\in \RR[t]^{n\times n}$, one can
find unimodular matrices $\pmat{U},\pmat{V} \in \RR[t]^{n\times n}$ such
that $\pmat{U}\A\pmat{V}$ is a diagonal matrix $\pmat{S}$.  Unimodular
means that there is a polynomial inverse matrix, or equivalently, that the
determinant is a nonzero constant from $\RR$.  The unimodular matrices
$\pmat{U},\pmat{V}$ encapsulate the polynomial row and column operations,
respectively, needed for such a diagonalization.  The best known
diagonalization is the Smith Normal Form (SNF, or simply Smith form) of a
matrix polynomial. Here
\[
  \pmat{S} =
  \begin{pmatrix}
    s_1 \\
    & s_2 \\
    && \ddots \\
    &&& s_n
  \end{pmatrix}\in\RR[t]^\nxn,
\]
where $s_1,\ldots,s_n\in\RR[t]$ are monic and $s_i\divs s_{i+1}$ for
$1\leq i<n$.  The Smith form always exists and is unique though the
associated unimodular matrices $\pmat{U}$, $\pmat{V}$ are not unique (see,
e.g., \citep{Kai80,GolLanRod09}).  The diagonal entries $s_1,\ldots,s_n$
are referred to as the \emph{invariant factors} of~$\A$.

Matrix polynomials and their Smith forms are used in many areas of
computational algebra, control systems theory, differential equations and
mechanics.  The Smith form is important as it effectively reveals the
structure of the polynomial lattice of rows and columns, as well as the
effects of localizing at individual eigenvalues.  That is, it characterizes
how the rank decreases as the variable $t$ is set to different values
(especially eigenvalues, where the rank drops).
The Smith form is closely related to the more general {\em Smith-McMillan
  form} for matrices of rational functions, a form that reveals the
structure of the eigenvalue at infinity, or the infinite spectral structure.  
The eigenvalue at infinity is non-trivial if the leading coefficient matrix is 
rank deficient or equivalently, the determinant does not achieve the generic 
degree.

The algebra of matrix polynomials is typically described assuming that the
coefficients are fixed and come from an exact arithmetic domain, usually
the field of real or complex numbers.  In this exact setting, computing
the Smith form has been well studied, and very efficient procedures are
available (see \citep{KalSto15} and the references
therein). However, in some applications, particularly control theory and mechanics,
the coefficients can come from measured data or contain some amount of
uncertainty.  Compounding this, for efficiency reasons such computations
are usually performed using floating point to approximate computations in
$\RR$, introducing roundoff error.
As such, algorithms must accommodate numerical inaccuracies and are prone
to numerical instability.

Numerical methods to compute the Smith form of a matrix
polynomial typically rely on linearization and orthogonal transformations
\citep{VanDew83,BeeVan88,DemKaa93,DemKaa93b,DemEde95} to infer the Smith
form of a nearby matrix polynomial via the Jordan blocks in the Kronecker
canonical form (see \citep{Kai80}).  These linearization techniques 
are numerically  
backwards stable, and for many problems this is sufficient to
ensure that the computed solutions are computationally useful when a
problem is continuous.

Unfortunately, the eigenvalues of a matrix polynomial are not necessarily
continuous functions of the coefficients of the matrix polynomial, and
backwards stability is not always sufficient to ensure computed solutions
are useful in the presence of discontinuities.  Previous methods are also
unstructured in the sense that the computed non-trivial Smith form may not
be the Smith form of a matrix polynomial with a prescribed coefficient
structure, for example, maintaining the degree of entries or not introducing
additional non-zero entries or coefficients. In extreme instances, the
unstructured backwards error can be arbitrarily small, while the structured
distance to an interesting Smith form is relatively large.
Finally, existing numerical methods can also fail to compute meaningful
results on some problems due to uncertainties. Examples of such problems
include nearly rank deficient matrix polynomials, repeated eigenvalues or
eigenvalues that are close together and other ill-posed instances.

In this paper we consider the problem of computing a nearby matrix
polynomial with a prescribed spectral structure, broadly speaking, the
degrees and multiplicities of the invariant factors in the Smith form, or
equivalently the structure and multiplicity of the eigenvalues of the
matrix polynomial.  The results presented in this paper extend those in the
conference paper \citep*{GieHarLab18}. This work is not so much about
computing the Smith normal form of a matrix polynomial using floating point
arithmetic, but rather our focus is on the computation of a nearby matrix
polynomial with ``an interesting'' or non-generic Smith normal form.  The
emphasis in this work is on the finite spectral structure of a matrix
polynomial, since the techniques described are easily generalized to
handle the instance of the infinite spectral structure as a special case.
\todo{made some changes here...}

The Smith form of a matrix polynomial is not continuous with respect to the 
usual topology of the coefficients and the resulting limitations of backward 
stability
is not the only issue that needs to be addressed when finding nearest
objects in an approximate arithmetic environment.
A second issue can be illustrated by recalling the well-known
representation of the invariant factors $s_1,\ldots,s_n$ of a matrix
$\A\in\RR[t]^\nxn$ as ratios $s_i=\delta_{i}/\delta_{i-1}$ of the
\emph{determinantal divisors}
$\delta_0, \delta_1,\ldots,\delta_n\in\RR[t]$, where
\[
  \delta_0 = 1, ~~\delta_i= \text{GCD}\biggl\{ \text{all $i\times i$ minors
    of $\A$} \biggr\} \in\RR[t].
\]
In the case of $2\times 2$ matrix polynomials, computing the nearest
non-trivial Smith form is thus equivalent to finding the nearest matrix
polynomial whose polynomial entries have a non-trivial
GCD. Recall that approximate GCD problems can have infima that are
\emph{unattainable}. That is, there are co-prime polynomials with nearby
polynomials with a non-trivial GCD at distances arbitrarily approaching an
infimum, while at the infimum itself the GCD is trivial (see, e.g.,
\citep*{GieHarKal19}).

The issue of unattainable infima extends to the Smith normal form. As an 
example, consider
\[
  \A = \left( \begin{array}{cc} t^2-2t+1 & \\ & t^2+2t+2 \end{array}
  \right) = \diag(f,g)\in\RR[t]^{2\times 2}.
\]
If we look for nearby polynomials $\ftil,\gtil\in\RR[t]$ of degree at most
$2$ such that $\gcd(\ftil,\gtil)=\gamma t+1$ (i.e. a nontrivial gcd) at
minimal distance $\norm{f-\ftil}_2^2+\norm{g-\gtil}_2^2$ for some
$\gamma\in\RR$, then it is shown in \cite[Example 3.3.6]{Har15} that this
distance is $(5\gamma^4-4\gamma^3+14 \gamma^2+2)/(\gamma^4+\gamma^2+1)$.
This distance has an infimum of $2$ at $\gamma=0$. However at $\gamma=0$ we
have $\gcd(\ftil,\gtil)=1$ even though $\deg(\gcd(\ftil,\gtil))>0$ for all
$\gamma\neq 0$. For $\A$ to have a non-trivial Smith form we must perturb
$f,g$ such that they have a non-trivial GCD, and thus any such perturbation
must be at a distance of at least $2$.  However, the perturbation of
distance precisely $2$ has a trivial Smith form.  There is no merit to
perturbing the off-diagonal entries of $\A$.

Our work indirectly involves measuring the sensitivity to the eigenvalues
of $\A$ and the determinant of $\A$. Thus we differ from most sensitivity
and perturbation analysis (e.g., \citep{Ste94,AhmAla09}),
since we also study how perturbations affect the invariant factors, instead
of the roots of the determinant.  Additionally our theory is able to
support the instance of $\A$ being rank deficient and having degree
exceeding one.  One may also approach the problem geometrically in the
context of manifolds \citep{EdeElmKaa97,EdeElmKaa99}. We do not consider
the manifold approach directly since it does not yield numerical
algorithms.

Determining what it means for a matrix polynomial to have a non-trivial
Smith form numerically and finding the distance from one matrix polynomial
to another matrix polynomial having an interesting or non-trivial Smith
form are formulated as finding solutions to continuous optimization
problems.  The main contributions of this paper are deciding when $\A$ has
an interesting Smith form, providing bounds on a ``radius of triviality''
around $\A$ and a structured stability analysis on iterative methods to
compute a structured matrix polynomial with desired spectral properties.

The remainder of the paper is organized as follows. In
Section~\ref{sec:prelim} we give the notation and terminology along with
some needed background used in our
work. Section~\ref{sec:trivial-snf} discusses the approximate Smith form
computation as an optimization problem and provides some new bounds on the
distance to non-triviality.  We present an optimization algorithm in
Section~\ref{sec:opt1-adjoint} with local stability properties and rapid
local convergence to compute a nearby matrix polynomial with a non-trivial
Smith form and discuss implementation details. A method to compute a matrix
polynomial with a prescribed lower bound on the number of ones in the Smith 
form is discussed
in Section~\ref{sec:optII-linearization}.
We discuss our implementation and include some examples in Section~\ref{sec:implementation}. The paper ends with a conclusion along with topics for future research.

A preliminary version of some of the results in this paper appears in the proceedings of the ISSAC
2018 conference \citep*{GieHarLab18}.

 \section{Preliminaries}
\label{sec:prelim}

In this section we give the notation and formal definitions needed to
precisely describe the problems summarized above.  We also present some
existing results used as building blocks for our work. In addition, we
provide a basic description of matrix functions and their first-order
derivatives (Jacobian matrices) which will be needed for the optimization
work central to our results.

\subsection{Notation and Terminology}
We make extensive use of the following terminology and definitions.  A
matrix polynomial $\A \in \RR[t]^\nxn$ is an $\nxn$ matrix whose entries
are polynomials. Typically we also work with matrices whose entries have
degree bound $d$ and use the notation $\RR_{\leq d}[t]^\nxn$ to describe
this set.  Alternatively, we may express matrix polynomials as
$\A = \sum_{1\leq j\leq d} \mat{A}_j t^j$ where $\mat{A}_j \in
\RR^\nxn$. The \emph{degree} of a matrix polynomial $d$ is defined to be
the degree of the highest-order non-zero entry of $\A$, or the largest
index $j$ such that $\mat{A}_j\neq 0$.  We say that $\A$ has \emph{full
  rank} or is \emph{regular} if $\det(\A)\nequiv 0$.  As noted in the
introduction, $\A$ is said to be \emph{unimodular} if
$\det(\A)\in \RR\backslash\{0\}$. The \emph{(finite) eigenvalues} are the
roots of $\det(A)\in\RR[t]$.  The norm of a polynomial $a\in \RR[t]$ is
defined as
$\norm{a} = \norm{a}_2 = \norm{(a_0,a_1,\ldots, a_d, 0,\ldots, 0)}_2$.  For
matrix polynomials we define
$\norm{\A} = \norm{\A}_F= \sqrt{\sum_{i,j} \norm{
    \pmat{A_{i,j}}}_2^2}$. Our choice of norm is a distributed coefficient
norm, sometimes known as the \emph{Frobenius norm}.

\begin{definition}[SVD -- \citealt{GolLoa12}]
The Singular Value Decomposition (SVD) of $A\in \RR^{n\times n}$ is given
by $U^T \Sigma V$, where $U,V \in \RR^{n\times n}$ satisfy $U^TU =I$,
$V^T V = I$ and $\Sigma = \diag(\sigma_1,\ldots, \sigma_{n})$ is a
diagonal matrix with non-negative real entries in descending order of
magnitude, the \emph{singular values} of $A$. The distance to the nearest
(unstructured) matrix of rank $m<n$ is $\sigma_{m+1}(A)$.
\end{definition}
For scalar matrices we frequently write $\norm{\cdot}_2$ for the largest
singular value, and $\sigma_{\min}(\cdot)$ for the smallest singular value.

In this paper we are mainly concerned with coefficient structures that
preserve the zero coefficient structure of a matrix polynomial, that is,
we generally do not change zero coefficients to non-zero, or increase the
degrees of matrix entries.
\begin{definition}[Affine/Linear Structure]
  A non-zero matrix polynomial $\A\in \R[t]^{n\times n}$ of degree at most
  $d$ has a \emph{linear structure} from a set $\mathcal{K}$ if
  $\A \in \Span(\mathcal{K})$ as a vector space over $\R$, where
  \[
    \mathcal{K} = \bigl\{C_{0,0},
\ldots, C_{0,k} ,t C_{1,0} ,
\ldots, tC_{1,k}, \ldots, t^d C_{d,0},\ldots, t^d C_{d,k}
    \bigr\},
  \]
  where $C_{l,j}\in \R^{n\times n}$ for $0\leq j \leq k$, where $k>0$ is a 
finite 
index variable.  If $\A=\mathcal{C}_0 + \mathcal{C}_1$,
  where $\mathcal{C}_0 \in \R[t]^{n\times n}$ is fixed and
  $\mathcal{C}_1 \in \Span(\mathcal{K})$, then $\A$ is said to have an
  \emph{affine} structure from the set~$\mathcal{K}$.
\end{definition}
\todo{I changed this. Removed LI and said $k$ was an index variable.}
Linearly and affine linearly structured matrices are best thought of as
imposing linear equality constraints on the entries. 
Examples of matrices with a linear structure include matrices with prescribed 
zero entries or coefficients, Toeplitz/Hankel matrices, Sylvester matrices, 
resultant-like matrices, Ruppert matrices and several other matrices appearing 
in symbolic-numeric computation. Matrices with  an affine structure include all 
matrices with a linear structure and, in addition, matrices having prescribed  
non-zero constant entries/coefficients, for example monic matrix polynomials.

Recall that the \emph{rank} of a matrix polynomial is the maximum number of
linearly independent rows or columns as a vector space over $\RR(t)$. This
is the same as the rank of the matrix $\A(\omega)$ for any $\omega\in\CC$ 
that is not an eigenvalue of $\A(t)$. If we allow evaluation at eigenvalues, 
then the McCoy rank is the lowest rank when $\A$ is evaluated at an eigenvalue.

\begin{definition}[McCoy Rank and Non-Trivial SNF]
  The McCoy rank of $\A$ is $\min_{\omega \in \CC} \{ \rank \A(\omega)\}$,
  the lowest rank possible when $\A$ is evaluated at any $\omega\in\CC$.
  Note that the rank of $\A$ only drops at all if it is evaluated at an
  eigenvalue $\omega\in\CC$.  The McCoy rank is also the number of ones in
  the Smith form. Equivalently, if $\A$ has $r$ non-trivial invariant
  factors, then the McCoy rank of $\A$ is $n-r$.  The matrix polynomial
  $\A$ is said to have a \emph{non-trivial or interesting Smith normal form} if 
the McCoy rank 
  is at most $n-2$, or equivalently, if it has two or more invariant factors of 
  non-zero degree. 
\end{definition}
\todo{ I changed this to allow zero as an invariant factor. It was wrong 
earlier.}

\begin{problem}[Approximate SNF Problem]
  \label{prb:approx-snf}
  Given a matrix polynomial $\A \in \RR[t]^\nxn$, find the distance to a
  non-trivial SNF. That is, find a matrix polynomial
  $\widehat {\A} \in \RR[t]^\nxn$ of prescribed coefficient structure that
  has a prescribed McCoy rank of at most $n-r$, for $r\geq 2$, such that
  $\norm{\A-\pmat{\widehat A}}$ is minimized under $\norm{\cdot}$, if such
  an $\widehat{\A}$ exists.
\end{problem}
We will consider $\norm{\cdot}=\norm{\cdot}_F$ to be the Frobenius norm.
The nearest matrix of McCoy rank at most $n-2$, if it exists, is called the
\emph{approximate SNF}.

\begin{problem}[Lower McCoy Rank Approximation Problem]
  \label {prb:lower-mccoy}
  Compute the nearest matrix with McCoy rank $n-r$ matrix, if it exists, for
  $r\geq 3$.
\end{problem}
In a generic sense, the nearest matrix polynomial with an interesting SNF
will have McCoy rank $n-2$ with probability one, but many matrices arising
from applications are expected to have more interesting (i.e. the invariant 
factors have a richer or non-generic multiplicity structure) Smith forms nearby.
\todo{made a change here}

As described in the introduction, it is possible that the distance
to a non-trivial SNF is not attainable. That is, there is a solution that
is approached asymptotically, but where the Smith form is trivial at the
infimum.
Fortunately, in most instances of interest, solutions will generally be
attainable. We will also later discuss how to identify and compute
unattainable solutions. Problem~\ref{prb:approx-snf} and
Problem~\ref{prb:lower-mccoy} admit the nearest rank $n-1$ or rank $n-2$
matrix polynomial as a special case. However, the computational challenges
are fundamentally different for non-trivial instances.

\subsection{Embedding into Scalar Domains}
In our study of nearest non-trivial Smith forms we often make use of the
representation of the diagonal elements as ratios of GCDs of
sub-determinants. When the coefficients are polynomials with numeric
coefficients it is helpful to embed the arithmetic operations of polynomial
multiplication and polynomial GCD into a matrix problem having numeric
coefficients (i.e., from $\RR$). Such an embedding allows us to employ a
number of tools, including condition numbers and perturbations of matrix
functions.

We start with some basic notation for mapping matrices and polynomials to
vectors.

\begin{definition}[Vec Operator]
  We define the operator $\vec:\RR[t]\to\RR^{(d+1)\times 1}$ as follows:
  \[
    p = \sum_{j=0}^d p_j t^t \in\RR[t] ~~\mapsto~~ \vec(p) =
    (p_0,p_1,\ldots, p_{d})^T \in \RR^{(d+1)\times 1}
  \]
The $\vec$ operator $\vec(\cdot)$ is extended to map a matrix 
$\RR[t]^{m\times n}$
  to a single vector in $\RR^{mn(d+1)\times 1}$ by stacking columns of
  (padded) coefficient vectors on top of each other.
  \[
    \A\in\RR[t]^{m\times n} \mapsto \vec(\A) = \begin{pmatrix}
      \vec(\A_{11})\\
      \vdots\\
      \vec(\A_{mn})
    \end{pmatrix}
    \in \RR^{mn(d+1)\times 1}.
  \]
\end{definition}

It is sometimes useful to reduce matrix polynomials to vectors of polynomials
in $\RR[t]$ rather than vectors over $\RR$.
\begin{definition}[Polynomial Vec Operator]
  The $\pvec$ operator maps $\RR[t]^{m\times n}$ to a vector
  $\RR[t]^{nm \times 1}$ as
  \[
    \A\in\RR[t]^{m\times n} \mapsto \pvec(\A) = \begin{pmatrix}
      \A_{11}\\
      \vdots\\
      \A_{mn}
    \end{pmatrix}
    \in\RR[t]^{mn\times 1}.
  \]
We define the vectorization of matrix polynomials in this somewhat
  non-standard way so that we can later facilitate the computation of derivatives of matrix
  polynomial valued functions.
\end{definition}

To describe polynomial multiplication in terms of linear maps over scalars we have:

\begin{definition}[Convolution Matrix]
  Polynomial multiplication between polynomials $a,b\in \RR[t]$, of degrees
  $d_1$ and $d_2$, respectively may be expressed as a Toeplitz-matrix-vector
  product.  We define
  \[
    \phi_{d_2}(a)
    =
    \begin{pmatrix}
          a_0 \\
          \vdots  & \ddots \\
          a_{d_1} & & a_0 \\
          & \ddots & \vdots \\
          & & a_{d_1}
        \end{pmatrix}
        \in\RR^{(d_1+d_2+1)\times (d_2+1)}.
    ~~~~~~~\mbox{It follows that}~~~
    \vec(ab)=\conv_{d_2}(a)\vec(b).
  \]
\end{definition}
When $a$ is non-zero, we can also define division through pseudo-inversion
or linear least squares. In a similar manner, we can define the product of 
matrix polynomials through a Toeplitz-block matrix. 

\begin{definition}[Block Convolution Matrix]
  We can express multiplication of a matrix and vector of polynomials,
  $\A \in \RR[t]^{m\times n}$ and $\bb \in \RR[t]^{n\times 1}$, of degrees
  at most $d_1$ and $d_2$ respectively, as a scalar linear system
  \[
    \vec(\A \bb) = \Phi_{d_2}(\A) \vec(\bb),
  \]
  where
  \[
    \Phi_{d_2}(\A) = \begin{pmatrix}
      \phi_{d_2}(\A_{11}) & \cdots &\phi_{d_2}(\A_{1n})\\
      \vdots & & \vdots \\
      \phi_{d_2}(\A_{m1}) & \cdots & \phi_{d_2}(\A_{mn})
    \end{pmatrix} \in \RR^{m(d_1+d_2+1) \times n(d_2+1)}.
  \]
\end{definition}
The block convolution matrix is sometimes referred to as a ``Sylvester
matrix'' associated with $\A$.  However, we reserve the term ``Sylvester
matrix'' for the more standard linear system appearing from the GCD of two
(or more) polynomials.  The block convolution matrix is a scalar matrix
whose entries have a linear (Toeplitz-block) structure.

\begin{definition}[Kronecker Product]
  The Kronecker product of $\A \in \RR[t]^{m\times n}$ and
  $\B \in \RR[t]^{k\times \ell}$ denoted as $\A\otimes \B$ is the
  $mk \times n\ell$ matrix over $\RR[t]$ defined as
  \[ \A \otimes \B = \begin{pmatrix}
      \A_{11} \B & \cdots &  \A_{1n} \B\\
      \vdots  &   & \vdots \\
      \A_{m1}\B & \cdots & \A_{mn} \B
    \end{pmatrix} \in \RR[t]^{mk \times n\ell}.
  \]
\end{definition}
This definition of Kronecker product, sometimes referred to as
the ``outer product'', also holds for scalar matrices (and vectors).

\begin{lemma}
  For scalar matrices of compatible dimension $A,X$ and $B$ over $\RR$, we
  can
  \[
    \vec(AXB) = (B^T\otimes A) \vec(X).
  \]
  Likewise, for matrix polynomials $\A, \pmat{X}$ and $\B$ of compatible
  dimension over $\RR[t]$, we have
  \[
    \pvec(\A \pmat{X} \B) = (\B^T \otimes \A) \pvec( \pmat{X}).
  \]
\end{lemma}
The Kronecker product can also be used to re-write matrix equations of the
form $AX=B$, for matrices $A$, $B$ and $X$ of compatible dimensions, to
\[
  \vec(AX) = (X^T\otimes I)\vec(A) = (I\otimes A)\vec(X) = \vec(B).
\]

\subsection{Derivatives of Matrix Polynomial Valued Functions}
\label{sec:dervs}

In this paper we will need to compute derivatives of some important matrix 
polynomial valued functions, namely the determinant and adjoint. 
This problem is approached in the context of  computing the Jacobian matrix of 
a vector valued function. The analysis in this section will be useful for 
showing that Lagrange multipliers typically exist in the optimization problems 
encountered. The quantities computed can also be used to derive first-order 
perturbation bounds for these matrix polynomial valued functions with respect 
to $\norm{\cdot}_F$.  

Recall that the adjoint of a matrix polynomial $\A\in\RR[t]^\nxn$, denoted by
$\Adj(\A)\in\RR[t]^\nxn$, is the transpose of the cofactor matrix. Thus 
$\Adj(\A)_{ij}=(-1)^{i+j}\det(\A[j|i])$ where
$\A[j|i]$ is the $(j,i)$ cofactor of $\A$, 
that is, the matrix formed by removing row $j$ and column $i$ from $\A$.
When $\A$ has full rank, $\A$ satisfies $\A \Adj(\A) = \det(\A)I$.   

The determinant of an $n \times n$ matrix polynomial having entries of degree 
at most $d$ can be viewed as a mapping from $\RR^{n^2 (d+1)} \to \RR^{nd + 1}$, 
since the determinant has degree at most $nd$. 
With this same viewpoint, we can view the adjoint of a matrix polynomial as a
mapping from $\RR^{n^2(d+1)} \to \RR^{n^2((n-1)d+1)}$, since the degree of the 
entries of the adjoint are at most $(n-1)d$. Our notation for computing 
derivatives of vector valued functions follows that of \citep{MagNeu88}.

It is not surprising that the determinant of a matrix
polynomial has a similar identity~\citep{MagNeu88} to the well-known scalar 
identity $\nabla \det(A) = \vec((\Adj(A)^T)^T).$ 

\begin{theorem}
  \label{thm:matrix-poly-adj-derivative}
  Let $\A \in \RR[t]^{n\times n}$ have degree at most $d$, then
  \[
    J_{\det} = \frac{\D \vec(\det(\A))}{\D \vec(\A)} \\
= \Phi_{d}(\pvec(\Adj(\A)^T)^T)\in \RR^{ (nd+1) \times n^2(d+1)}.
  \]
\end{theorem}
\begin{proof}
  We note that from generalizing the scalar identity
  $\nabla \det(\cdot) = \vec(\Adj(\cdot)^T)^T$, we can write a first-order
  expansion of the determinant as
  \[ \det(\A+\Delta \A) = \det(\A)+ \pvec(\Adj(\A)^T)^T \pvec(\Delta
    \A) + O(\norm{\Delta \A}_F^2),\] and ignoring
  higher-order terms we obtain the scalar expression
  \[ \vec(\det(\A+\Delta \A)) \approx \vec(\det(\A))+
    \vec(\pvec(\Adj(\A)^T)^T \pvec(\Delta \A)).\]

  The Jacobian can be extracted by (padding with zero coefficient entries
  as necessary) writing
  $ \vec(\pvec(\Adj(\A)^T) ^T \pvec(\Delta \A)) = J_{\det} \vec(\Delta \A)$
  as a matrix-vector product.  Thus, using block-convolution matrices we
  have
  \[ \frac{\D\vec(\det(\A))}{\D \vec(\A)} = \nabla(\det(\A)) =
    \Phi_{d}(\pvec(\Adj(\A)^T)^T). \qedhere \] 
\end{proof}

Now that we have a closed-form expression for the derivative of the 
determinant, it is useful to derive a closed-form expression for the adjoint 
matrix.  The closed-form expression reveals rank information, and is 
independently useful for optimization algorithms requiring derivatives. The 
rank information is useful to obtain insights about the existence of Lagrange 
multipliers. 
If $J_{\det}$ has full or locally constant (row) rank then constraint 
qualifications will hold for several constrained optimization problems 
involving the determinant. If $\Adj(\A)$ is non-zero then one can often infer 
the existence of Lagrange multipliers for other problems as well.

\begin{theorem}\label{thm:matrix-poly-adjoint-2nd-derivative}
  Let $\A\in \RR[t]^{n\times n}$ have degree at most $d$ and rank $n$.  The
  Jacobian of $\Adj(\A)$ is
  $J_{\Adj}\in \RR^{(n^2((n-1)d+1)) \times n^2(d+1)}$ with
  \[
    J_{\Adj} = \left[ \Phi_{(n-1)d} (I\otimes \A)\right] ^+ \left[
      \Phi_{d}(\pvec(I) \pvec(\Adj(\A)^T)^T) - \Phi_{d}(\Adj(\A)^T\otimes
      I)\right],
  \]
  where $I$ is understood to be the $n\times n$ identity matrix and 
for a scalar matrix $A$ of full rank, $A^+$ is the Moore-Penrose pseudo-inverse 
arising from the SVD. 
\end{theorem}

\begin{proof}
  First recall that if $\A$ has full rank, then
  $\A\Adj(\A) = \Adj(\A)\A = \det(\A)I$. This expression defines the
  adjoint matrix when $\A$ has full rank.  We can write
  \[
    \pvec(\A \Adj(\A)) = (\Adj(\A)^T \otimes I)\pvec(\A) = (I\otimes \A)
    \pvec(\Adj(\A)),
  \]
  thus converting to a linear system over $\RR$ produces
  \[
    \vec(\A\Adj(\A)) = \Phi_{(n-1)d}(I\otimes \A)
    \vec(\Adj(\A)) = \Phi_{d}(\Adj(\A)^T\otimes I)\vec(\A).
  \]
  Applying the product rule yields
  \begin{equation}
    \label{eqn:prod-rule}
    \D \vec(\A \Adj(\A)) =
    (\D\Phi_{(n-1)d}(I \otimes \A)) \vec(\Adj(\A)) + 
    \Phi_{(n-1)d}(I\otimes \A) \D \vec(\Adj(\A)).
  \end{equation}
  Next we observe that \eqref{eqn:prod-rule} has the same coefficients as
  the expression
  \[
    \vec( (\D \A) \Adj(\A) + \A (\D \Adj(\A)) )
  \]
  which is equivalent to
  \[
    \vec( (\Adj(\A)^T\otimes I) \pvec(\D \A) + (I\otimes \A) \pvec(\D
    \Adj(\A))),
  \]
  which reduces to
  \begin{equation}
    \label{eqn:lhs-final}
    \Phi_{d}((\Adj(\A)^T\otimes I)) \vec(\D \A) + 
    \Phi_{(n-1)d}(I\otimes A) 
    \vec( \D \Adj(\A)).
  \end{equation}
  We now have the derivative of the left hand side the expression
  $\A\Adj(\A)=\det(\A)I$.  Differentiation of the right hand side
  yields
  \[
    \D \vec(\det(\A)I) = \vec( \D \pvec(\det(\A) I)),
  \]
  which is equivalent to the expression
  \begin{equation}
    \label{eqn:rhs-derivative}
    \vec( \D \pvec(\det(\A) I)) =
    \vec( \pvec(I) \pvec(\Adj(\A)^T)^T \pvec(\D\A)).
  \end{equation}
  Converting \eqref{eqn:rhs-derivative} into a linear system over $\RR$
  leads to
 
  \begin{equation}
    \label{eqn:rhs-final}
    \vec( \pvec(I) \pvec(\Adj(\A)^T)^T) \pvec(\D \A) = \Phi_{d}(\pvec(I)
    \pvec(\Adj(\A)^T)^T) \vec(\D \A),
  \end{equation}
  which is the derivative of the right-hand side.

  Combining \eqref{eqn:lhs-final} and \eqref{eqn:rhs-final} we have
  \[
    \Phi_{(n-1)d} (I\otimes \A) \frac{\D\vec(\Adj(\A))}{\D \vec(\A)} =
    \Phi_{d}(\pvec(I) \pvec(\Adj(\A)^T)^T) - \Phi_{d}(\Adj(\A)^T\otimes I).
  \]

  Assuming that $\A$ has full rank so
  $\Phi_{(n-1)d} (\pvec(I\otimes \A) ) $ is pseudo-invertible, we can write
  \[
    J_{\Adj} = \left[ \Phi_{(n-1)d} (I\otimes \A) \right] ^+ \left[
      \Phi_{d}(\pvec(I) \pvec(\Adj(\A)^T)^T) - \Phi_{d}(\Adj(\A)^T\otimes
      I)\right],
  \]
  which completes the proof.
\end{proof}
An observation that is important later is that the derivative of the
adjoint has a Toeplitz-block structure.  More importantly, the bandwidth is
$O(d)$, and we only need to compute $O(n^2)$ columns instead of $O(n^2d)$.
We also note that $J_{\Adj}$ may be padded with zeros, since $\A$ may not
have generic degrees.

\begin{corollary}
  \label{cor:adjoint-derivative-rank}
  If $\A$ has full rank then $J_{\Adj}$ has full rank.
\end{corollary}
\begin{proof}
  The matrix $\Phi_{(n-1)d}(I\otimes \A)$ has full rank since $I\otimes \A$
  has full rank.  The matrix
  \begin{equation}
    \label{eqn:adjoint-rank1-update}
    \pvec(I)\pvec(\Adj(\A)^T)^T - \Adj(\A)^T\otimes I  = -\left( 
      -\pvec(I)\pvec(\Adj(\A)^T)^T + \Adj(\A)^T\otimes I  \right)
  \end{equation}
  is a rank one update to a matrix polynomial.  By evaluating
  \eqref{eqn:adjoint-rank1-update} at a complex number $\omega$ that is not
  an eigenvalue of $\A$ we can show that \eqref{eqn:adjoint-rank1-update}
  has full rank. Let $A = \A(\omega)$, so $A \in \RR^{n\times n}$ has full
  rank.

  Using the Sherman-Morrison formula ~\citep[pg.~487]{Hig02} for rank 1
  updates to a matrix, we need to verify that
  \[
    1- \vec( \Adj(A)^T)^T \left[\left( \Adj(A)^T\right)^{-1} \otimes
      I\right]\vec(I) \neq 0,
  \]
in order to ensure that \eqref{eqn:adjoint-rank1-update} has full rank.
  We have that
  \begin{align*}
    \vec( \Adj(A)^T)^T  \left[\left( \Adj(A)^T\right)^{-1} \otimes 
    I\right]\vec(I) &= \vec(\Adj(A)^T)^T \vec\left( \Adj(A)^T)^{-1} \right) \\
                    &=  \Tr \left(\Adj(A)^T \left(\Adj(A)^T \right)^{-1} \right)\\
                    &= n,
  \end{align*}
  thus \eqref{eqn:adjoint-rank1-update} has full rank.  Note we used the
  identities for matrices $X, Y$ and $Z$ of appropriate dimension, that
  $\vec(XYZ) = (Z^T \otimes X) \vec(Y)$ and
  $\vec(X^T)^T \vec(Y) = \Tr(XY)$.  Again, we have that
  \[ \Phi_{d}(\pvec(I) \pvec(\Adj(\A)^T)^T) - \Phi_{d}(\Adj(\A)^T\otimes
    I)\] has full rank, thus $J_{\Adj}$ is a product of two matrices of
  full rank, so $J_{\Adj}$ must also have full rank.
\end{proof}

Corollary~\ref{cor:adjoint-derivative-rank} implies that Lagrange multipliers 
will exist to several optimization problems involving the adjoint matrix as a 
constraint, since the Jacobian matrix  of the adjoint has full rank.
The linear independent constraint qualification or the constant rank constraint
qualification will hold for several optimization problems of the form
\[
  \min \norm{\Delta \A} \text{~~subject to~~} \Adj(\A+\Delta \A) = \pmat{F},
\]
for some reasonably prescribed $\pmat{F} \in \RR[t]^{n\times n}$.

\begin{remark}
  If $\A$ is rank deficient, then the derivative is still defined, but not
  necessarily by Theorem~\ref{thm:matrix-poly-adjoint-2nd-derivative}.  If
  $\rank(\A) \leq n-3$ then $J_{\Adj} =0$, since all $(n-3) \times (n-3)$
  minors vanish ($J_{\Adj}$ consists of the coefficients of these
  minors). If $\rank(\A)= n-1$ or $\rank(\A) = n-2$ then $J_{\Adj}$ is
  still defined and in both cases $J_{\Adj}\neq 0$. However $J_{\Adj}$ is
  not necessarily described by
  Theorem~\ref{thm:matrix-poly-adjoint-2nd-derivative}.
\end{remark} 
For several affine or linear perturbation structures (such as ones that
preserve the degree of entries or the support of entries),
Theorem~\ref{thm:matrix-poly-adjoint-2nd-derivative} and the associated
Corollary~\ref{cor:adjoint-derivative-rank} will hold (after deleting some
extraneous rows or columns).

 \section{When Does a Numerical Matrix Polynomial have a trivial SNF?}
\label{sec:trivial-snf}

In this section we consider the question of determining if a matrix
polynomial has a non-trivial SNF, or rather how much do the coefficients
need to be perturbed to have a non-trivial SNF.  We provide a lower bound
on this distance by analyzing the distance to a reduced-rank generalized
Sylvester matrix.

\subsection{Embeddings into generalized Sylvester matrices and approximate GCDs}
In the introduction we demonstrated that some nearby non-trivial Smith Forms are
unattainable. In this subsection we investigate why these unattainable
values occur. We first review some basic results needed to analyze the
topology of the approximate Smith form problem.

For a matrix $\A\in\RR[x]^\nxn$, we know that
$s_{n}=\delta_{n}/\delta_{n-1}$, the quotient of the determinant and the
GCD of all $(n-1)\times (n-1)$ minors.  Since these minors are precisely
the entries of the adjoint matrix, it follows that $\A$ has a non-trivial
Smith form if and only if the GCD of all entries of the adjoint is
non-trivial, that is, $\deg (\gcd(\{\Adj(\A)_{ij} \} ) )\geq 1$. In order
to obtain bounds on the distance to a matrix having a non-trivial Smith
form, we consider an approximate GCD problem of the form
\[
  \min \left\{\norm{\Delta \A} \text{~~subject to~~} \deg
    \left(\gcd\left\{\Adj\left(\A+\Delta \A\right)_{ij} \right)\right\}
    \neq 1\right\}.
\]\todo{I changed this.}
If this was a classical approximate GCD problem, then the use of
Sylvester-like matrices would be sufficient. However, in our problem the
degrees of the entries of the adjoint may change under perturbations. In
order to perform an analysis, we need to study a family of generalized
Sylvester matrices that allow higher-degree zero coefficients to be
perturbed.

The computation of the GCD of many polynomials is typically embedded into a
scalar matrix problem using the classical Sylvester matrix. However, in our
case we want to look at GCDs of nearby polynomials but with the added
wrinkle that the degrees of the entries of the individual polynomials may
change under perturbations. In order to perform such an analysis, we need
to study a family of generalized Sylvester matrices that allow
higher-degree zero coefficients to be perturbed.

Let $\f = (f_1,\ldots,$ $f_k) \in \RR[t]^k$ be a vector of polynomials with
degrees $\dd = (d_1,\ldots, d_k)$ ordered as $d_j \geq d_{j+1}$ \todo{ fixed 
the order. We wanted the mto be decreasing. oops. }for
$1\leq j \leq k-1$.  Set $d=d_1$ and $\ell = \max(d_2,\ldots, d_k)$ and
suppose that for each $i\in\{2,\ldots, k\}$ we have
$f_i = \sum_{1\leq j\leq \ell} f_{ij} t^j$.
\begin{definition}[Generalized Sylvester Matrix]
  The \emph{generalized Sylvester matrix} of ${\bf{f}}$ is defined as
  \[
    \Syl(\f) = \Syl_{\dd}(\f) = \begin{pmatrix}
      \conv_{\ell-1}(f_1)^T \\
      \conv_{d-1}(f_2)^T \\
\vdots  \\
      \conv_{d-1}(f_k) ^T
    \end{pmatrix}\in \RR^{ (\ell + (k-1)d) \times (\ell+d)}.
  \]
\end{definition}

Some authors, e.g., \citep{FatKar03,VarSto78}, refer to such a matrix as an
expanded Sylvester matrix or generalized resultant matrix. The generalized
Sylvester matrix has many useful properties pertaining to the B\'ezout
coefficients. However, we are only concerned with the well known result
that $\gcd(\f)=\gcd(f_1,\ldots,f_k)=1$ if and only if $\Syl_{\dd}(\f)$ has full rank
\citep{VarSto78}.

Sometimes  treatreating a polynomial of degree $d$ as one of
larger degree is useful.
This can be accomplished by constructing a similar matrix
and padding rows and columns with zero entries.  The generalized Sylvester
matrix of degree at most $\dd^\prime\geq \dd$ (component-wise) of $\f$ is
defined analogously as $\Syl_{\dd^\prime}(\f)$, taking $d$ to be the
largest degree entry and $\ell$ to be the largest degree of the remaining
entries of $\dd^\prime$. Note that $\ell=d$ is possible and typical. If the
entries of $\f$ have a non-trivial GCD (that is possibly unattainable)
under a perturbation structure $\Delta \f$, then it is necessary that
$\Syl_{\dd^\prime}(\f)$ is rank deficient, and often this will be
sufficient.

If we view the entries of $\f$ as polynomials of degree $\dd^\prime$ and
$d^\prime_i > d_i$ for all $i$, then the entries of $\f$ have an
unattainable GCD of distance zero, typically of the form
$1+\varepsilon t \sim t+\varepsilon^{-1}$. In other words, the underlying
approximate GCD problem is ill-posed in a sense that the solution is 
unattainable. In order to study the theory of unattainable GCD's, sometimes 
referred to as GCD's at infinity, we need to study the notion of a degree 
reversed polynomial.

\begin{lemma}
  \label{lem:sylvester-finite-detection-strong}
  If $\max(\dd) = \max(\dd^\prime)$ then $\Syl_{\dd}(\f)$ has full rank if
  and only if and $\Syl_{\dd^\prime}(\f)$ has full rank.
\end{lemma}
\begin{proof}
  Let $d$ and $\ell$ be the largest and second largest entries of $\dd$ and
  $\ell^\prime$ be the second largest entry of $\dd^\prime$.  The result
  follows from the main theorem of \cite{VarSto78} by considering the case
  of $\ell^\prime = d$.
\end{proof}

This lemma characterizes the (generic) case when elements of maximal degree
of
$\f$ do not change under perturbations, in which case the generalized Sylvester
matrix still meaningfully encodes GCD information.  However, it is possible
that $\Syl_{\dd}(\f)$ has full rank and
$\Syl_{\dd^\prime}(\f)$ is rank deficient but the distance to a non-trivial
GCD is not zero.  This can occur when $d_j = d^\prime_j$ for some
$j$ and $\dd^\prime \geq
\dd$. To understand the most general case, we need to look at generalized
Sylvester matrices of the reversal of several polynomials.

\begin{definition}
  The \emph{degree $d$ reversal} of $f\in \RR[t]$ of degree at most $d$
  is defined as $\rev_d(f)= t^df(t^{-1})$. For a vector of polynomials
  $\f \in \RR[t]^{k}$ of degrees at most $\dd=(d_1,\ldots,d_k)$ the
  \emph{degree $\dd$ reversal} of $\f$ is the vector
  $\rev_{\dd}(\f) = (\rev_{d_1} (f_1),\ldots, \rev_{d_k} (f_k))$.
\end{definition}

The following theorem enables us to determine if unattainable solutions are
occurring in an approximate GCD problem with an arbitrary (possibly
non-linear) structure on the coefficients.

\begin{theorem}
  \label{thm:sylvester-reversal}
Let $\f$ be a vector of non-zero polynomials of degree at most $d$.
  Suppose that $\Syl_\dd(\f)$ has full rank and $\Syl_{\dd^\prime}(\f)$ is
  rank deficient, where the perturbations $\Delta \f$ have degrees at most
  $\dd^\prime$ and the entries of $\f$ have degrees $\dd$.  Then $\f$ has
  an {\emph{unattainable}} non-trivial GCD of distance zero under the
  perturbation structure $\Delta \f$ if and only if
  $ \Syl (\rev_{\dd^\prime} (\f))$ is rank deficient.
\end{theorem} 
\begin{proof}
  Suppose that $\Syl( \rev_{\dd^\prime} (\f))$ has full rank.  Then
  $\gcd(\rev_{\dd^\prime}(\f))=1$,
hence $\f$ does not have an unattainable non-trivial GCD, since
  $\gcd(\f) =1$.  Conversely, suppose that $\Syl( \rev_{\dd^\prime} (\f))$
  is rank deficient. Then, $t$ is a factor of
  $\gcd( \rev_{\dd^\prime}(\f))$ but $t$ is not a factor of
  $\gcd(\rev_\dd (\f))$. Accordingly, all entries of
  $\f+\Delta \f$ may increase in degree and so the distance of $\f$ having
  a non-trivial GCD is zero, and so is unattainable.
\end{proof}

If the generalized Sylvester matrix of $\f$ has full rank, but the
generalized Sylvester matrix that encodes the perturbations $\f+\Delta \f$
is rank deficient, then either there is an unattainable GCD, or the
generalized Sylvester matrix is rank deficient due to over-padding with
zeros.  Theorem~\ref{thm:sylvester-reversal} provides a reliable way to
detect this over-padding.

\todo{We need to say how we handle structured perturbations}
\begin{definition}
  We say that $\A$ has an \emph{unattainable non-trivial Smith form} if
  $\gcd (\Adj(\A))=1$ and $\gcd(\Adj(\A+\widetilde \Delta \A))\neq 1$ for
  an arbitrarily small perturbation
$\widetilde \Delta \A =\widetilde\Delta (\Delta \A)$ of some prescribed 
affine structure.
\end{definition}
Note that $\widetilde \Delta \A$ just means that perturbations to $\A$ are 
structured as an affine function of $\Delta \A$. \todo{I changed this shit or 
something.}
It is important
to carefully consider structured perturbations, because some matrix polynomials have an
unattainable non-trivial SNF under unstructured perturbations, but have an
attainable non-trivial SNF under structured perturbations (perturbations
that preserve the degree of entries or support of entries are structured).
Solutions that cannot be attained correspond to an eigenvalue at infinity
of $\A$ with a non-trivial spectral structure.  Such examples are easily
constructed when $\det(\A)$ or $\Adj(\A)$ have non-generic degrees.

\begin{example}
  \label{ex:unattainable}
  Let \[\A = \begin{pmatrix}
      t & t-1 \\
      t+1 & t
    \end{pmatrix} \in \RR[t]^{2\times 2} \text{~~and~~} \pmat{C} =
    \begin{pmatrix}
      \A \\
      & \A
    \end{pmatrix} \in \RR[t]^{4\times 4}.\] Then $\pmat{C}$ has an unattainable 
non-trivial
  Smith form if all perturbations to $\A$ are support or degree preserving
  (i.e. they preserve zero entries or do not increase the degree of each 
entry), both
  linear structures. Note that $\A$ and $\pmat{C}$ are both
  unimodular. However small perturbations to the non-zero coefficients of
  $\A$ make $\A+\Delta \A$ non-unimodular.
  
 The Smith
  form of $\rev(\pmat{C}) = t \pmat{C}|_{t=t^{-1}}$ is
  \[ \SNF(\rev(\pmat{C})) = \begin{pmatrix}
      1 \\
      & 1 \\
      & & t^2 \\
      & & & t^2
    \end{pmatrix},
  \]
which   implies that the eigenvalue at infinity of $\A$ has a non-trivial
  spectral structure. The eigenvalue at infinity having a non-trivial
  spectral structure implies that the SNF of $\pmat{C}$ is unattainable.
  Note that this is equivalent to saying that $\pmat{C}$ has a non-trivial 
Smith-McMillan form.
\end{example}

These examples are non-generic. Generically, the degree of all entries in
the adjoint will be $(n-1)d$ and will remain unchanged locally under
perturbations to the coefficients. Computing the distance
to the nearest matrix polynomial with a non-trivial Smith form under a
prescribed perturbation structure can be formulated as finding the nearest rank deficient
(structured) generalized Sylvester matrix of the adjoint or the
$\dd^\prime$ reversal of the adjoint.

\subsection{Nearest Rank Deficient Structured Generalized Sylvester Matrix}
Suppose that $\A \in \RR[t]^{n\times n}$ of degree at most $d$ has a
trivial Smith form and does not have an unattainable non-trivial Smith
form.  Then one method to compute a lower bound on the distance the entries
of $\A$ need to be perturbed to have an attainable or unattainable
non-trivial Smith form is to solve
\begin{equation}
  \label{eqn:generalized-sylvester-opt}
  \inf \norm{\Delta \A} \text{~~subject to~~} 
  \begin{cases} 
    ~\rank(\Syl_{\dd^\prime}( \Adj(\A+ \widetilde \Delta \A)))  <e,\\
    ~e=\rank(\Syl_{\dd^\prime}(\Adj(\A))).
  \end{cases}
\end{equation}
Here $\dd^\prime$ is the vector of the largest possible degrees of each entry
of $\Adj(\A+\widetilde \Delta \A)$, and $\widetilde \Delta \A)$ is a prescribed 
linear or
affine perturbation structure.
\todo{I made a change here... typo was fixed as well. }

It is sufficient to compute $\max (\dd^\prime)$, a quantity which will  generically be $(n-1)d$. For non-generic instances we require the
computation of $\dd^\prime$.  This optimization problem is non-convex, but
multi-linear in each coefficient of $\Delta \A$.

We do not attempt to solve this problem directly via numerical techniques,
since it enforces a necessary condition that is often sufficient.  Instead
we use it to develop a theory of solutions which can be exploited by faster
and more robust numerical methods.

\begin{lemma}\label{lem:gen-sylvester-open-set}
  Let $\f$ be a vector of polynomials with degrees $\dd$ and admissible
  perturbations $\Delta \f$ of degrees $\dd^\prime$ where
  $\max(\dd) \leq \max(\dd^\prime)$. Then the family of generalized
  Sylvester matrices $\Syl_{\dd^\prime}(\f)$ of rank at least $e$ form an
  open set under the perturbations $\Delta \f$.
\end{lemma}
\begin{proof}
  By the degree assumption on $\Delta \f$ we have that for an infinitesimal
  $\Delta \f$ that $\Syl_{\dd^\prime}( \f)$ and
  $\Syl_{\dd^\prime}( \Delta \f)$ have the same dimension.  Accordingly,
  let us suppose that $\Syl_{\dd^\prime}( \f)$ has rank at least $e$.  Then
  the Sylvester matrix in question must have rank at least $e$ in an 
open-neighborhood around it. In  particular, when
  $\norm{\Syl_{\dd^\prime}(\Delta \f)}_2 < \sigma_{e}(\Syl_{\dd^\prime}(
  \f))$ then
  $\rank (\Syl_{\dd^\prime} (\f+\Delta \f)) \geq \rank (\Syl_{\dd^\prime}(
  \f))$ and the result follows.
\end{proof}

\begin{theorem}\label{thm:distance-well-posed}
  The optimization problem \eqref{eqn:generalized-sylvester-opt} has an
  attainable global minimum under linear perturbation structures.
\end{theorem}
\begin{proof}
  Let $\mathcal{S}$ be the set of all rank at most $e-1$ generalized
  Sylvester matrices of prescribed shape by $\dd^\prime$ and $ \Adj(\A)$.
  Lemma~\ref{lem:gen-sylvester-open-set} implies that $\mathcal{S}$ is
  topologically closed.

  Let
  $\mathcal{R} = \{ \Syl_{\dd^\prime}( \Adj(\pmat{C})) \text{ subject to
  }\norm{\pmat{C}} \leq \norm{\A}\}$, where the generalized Sylvester
  matrices are padded with zeros to have the appropriate dimension if
  required. Since $\Delta \A$ has a linear perturbation structure, a
  feasible point is always $\mat{C} = -\A$.  By inspection $\mathcal{R}$ is
  seen to be a non-empty set that is bounded and closed.

  The functional $\norm{\cdot}$ is continuous over the non-empty closed and
  bounded set $\mathcal{S}\cap\mathcal{R}$.  Let
  $\B \in \mathcal{S}\cap\mathcal{R}$.  By Weierstrass's theorem
  $\norm{ \A -\pmat{B}}$ has an attainable global minimum over
  $\mathcal{S}\cap\mathcal{R}$.
\end{proof}
Note that if a feasible point exists under an affine perturbation
structure, then a solution to the optimization problem exists as well.
What this result says is that computing the distance to non-triviality is
generally a well-posed problem, even though 
computing a matrix polynomial of minimum distance may be ill-posed (the 
solution is unattainable).
The same results also hold when working over the $\dd^\prime$ reversed
coefficients.
A similar argument is employed by \citep[Theorem 2.1]{KYZ07}.

\subsection{Bounds on the Distance to non-triviality}
Suppose that $\A \in \RR[t]^{n\times n}$, of degree at most $d$, has a
trivial Smith form and does not have an unattainable non-trivial Smith
form.  This section provides some basic bounds on the distance coefficients
of $\A$ need to be perturbed to have a non-trivial Smith form. The bounds we 
derive are unstructured, although they can be generalized to several 
perturbation structures (such as ones that preserve the degree or support of 
entries) in a straight forward manner.

If we consider the mapping $\Adj(\cdot)$ as a vector-valued function from
$\RR^{n^2(d+1)} \to \RR^{n^2((n-1)d+1)}$ (with some coordinates possibly
fixed to zero), then we note that the mapping is locally Lipschitz. More
precisely, there exists $c>0$ such that for a sufficiently small $\Delta \A$,
\[\norm{ \Adj(\A) - \Adj(\A+\Delta \A)} \leq c \norm{\Delta \A}.\] The
quantity $c$ can be approximately bounded above by the (scalar) Jacobian matrix
$\nabla \Adj(\cdot)$ evaluated at $\A$.  A local upper bound for $c$ is
approximately $\norm{\nabla \Adj(\A)}_2$.  We can invoke
Theorem~\ref{thm:matrix-poly-adjoint-2nd-derivative} if $\A$ has full rank. By 
considering $\hat{c} = \tallnorm{\left[ \Phi_{(n-1)d} (I\otimes \A )\right] 
^+}_2$, we
obtain the (absolute) first-order {\emph{approximate}} perturbation bound
\[
  \norm{\Adj(\A) - \Adj(\A+\Delta \A)}_F
  \lesssim \hat{c}(n+\sqrt{n}) (d+1)\norm{\Adj(\A)}_F \norm{\Delta \A }_F.
\]

The entries of $\nabla \Adj(\A)$ consist of the coefficients of the
$(n-2)\times (n-2)$ minors of $\A$.  This follows because $\Adj(\cdot)$ is
a multi-linear vector mapping and the derivative of each entry is a
coefficient of the leading coefficient with respect to the variable of
differentiation.  The size of each minor can be bounded above (albeit poorly)
by Hadamard's inequality (Goldstein-Graham variant, see \citep{Los74}). As
such, we have the sequence of bounds
\begin{align*}
  \norm{\nabla \Adj(\A)}_2
  ~\leq~ n\sqrt{d+1} \norm{\nabla \Adj(\A)}_\infty 
  ~\leq~ n^3 (d+1)^{5/2} \norm{\A}_\infty^{n-2}(d+1)^{n-2} n^{(n-2)/2}, 
\end{align*}
where $\norm{\A}_\infty$ is understood to be a vector norm and
$\norm{\nabla \Adj(\A)}_\infty$ is understood to be a matrix norm. The
bound in question
can be used in conjunction with the SVD to obtain a lower bound on the
distance to a matrix polynomial with a non-trivial Smith form.

\begin{theorem}
  Suppose that $\dd^\prime = (\gamma,\gamma\ldots, \gamma)$ and
  $\Syl_{\dd^\prime}( \Adj(\A))$ has rank $e$. Then an approximate lower bound 
on the
  distance to non-triviality is
  \[ \frac{1}{\gamma\norm{\nabla
        \Adj(\A)}}_F\sigma_{e}(\Syl_{\dd^\prime}(\Adj(\A))). \]
\end{theorem}
\begin{proof}
  We note that for polynomials $\f$ with degrees $\dd^\prime$ that
  $\norm{ \Syl_{\dd^\prime}( \f)}_F = \gamma \norm{\f}_F$.  Accordingly, if
  $\Delta \A$ is a minimal perturbation to non-triviality, then
  \begin{align*}
    \frac{1}{\gamma}\sigma_{e}(\Syl_{\dd^\prime}( \Adj(\A))) & 
                                                               \leq \norm{\Adj(\A) 
                                                               -\Adj(\A + 
                                                               \Delta \A) }_F \\
                                                             &\lesssim \norm{ 
\nabla \Adj(\A) }_F \norm{\Delta \A}_F,
  \end{align*}
  and the theorem follows by a simple rearrangement. Note that
  $\norm{\cdot}_2 \leq \norm{\cdot}_F$.
\end{proof}
If $\dd^\prime$ has different entries, then
$\ell \norm{\f}_F \leq \norm{\Syl_{\dd^\prime}( \f)}_F \leq \gamma \norm{\f}_F,$
where $\gamma$ and $\ell$ are the largest and second-largest entries of
$\dd^\prime$.  The lower bound provided can also be improved using the
Karmarkar-Lakshman distance \citep{KarLak96} in lieu of the smallest
singular value of the generalized Sylvester matrix, the $\dd^\prime$
reversal of the adjoint or other approximate GCD lower
bounds (e.g., \citep{BecLab98}).

 \section{Approximate SNF via Optimization}
\label{sec:opt1-adjoint}

In this section we formulate the approximate Smith form problem as the
solution to a continuous constrained optimization problem. We assume that
the solutions in question are attainable and develop a method with rapid
local convergence.  As the problem is non-convex, our convergence analysis
will be local.

\subsection{Constrained Optimization Formulation}

An equivalent statement to $\A$ having a non-trivial attainable Smith form
is that $\Adj(\A) = \Fstar h$ where $\Fstar$ is a vector (or matrix) of scalar
polynomials and $h$ is a divisor of $\gcd(\Adj(\A))$.  This directly leads
to the following optimization problem:
\begin{equation}
  \label{eqn:opt-prob}
  \min \norm{\Delta \pmat{A}}_F^2 \text{~~subject to~~}
  \begin{cases}
    \Adj(\pmat{A}+\Delta \pmat{A}) =\Fstar h, & \Fstar
    \mskip-5mu\in\mskip-3mu \RR[t]^{n\times
      n}, h\in \RR[t], \\
    \mathcal{N}_h\vec(h) =1, &\mathcal{N}_h \in \RR^{ 1 \times (\deg (h) +1)
    }.
  \end{cases}
\end{equation}
This is a multi-linearly structured approximate GCD problem which is a
non-convex optimization problem. Instead of finding a rank deficient
Sylvester matrix, we directly enforce that the entries of $\Adj(\pmat{A})$
have a non-trivial GCD.  The normalization requirement that
$\mathcal{N}_h \vec (h)=1$ is chosen to force $h$ to have a non-zero
degree, so that $h$ is not a scalar.  One useful normalization is to define
$\mathcal{N}_h$ such that $\lcoeff(h)=1$  (that is $\lcoeff(\cdot)$ is the 
leading coefficient of a polynomial). Explicitly, we assume the degree of the
approximate GCD is known and make it monic. Of course, other valid
normalizations also exist.

Since we are working over $\RR[t]$, there will always be a quadratic,
linear or zero factor of attainable solutions. 
If $h=0$ then the approximate SNF of $\pmat{A}$ is rank
deficient and computing approximate SNF reduces to the nearest rank at-most
$n-1$ or $n-2$ matrix polynomial problems, both of which are
well-understood \citep*{GieHarLab17,GieHarLab19}. Assuming that we are
now working in the nonzero case, we can assume generically that $\deg (h) =1$ 
or $\deg (h)=2$.

\subsection{Lagrange Multipliers and Optimality Conditions}
In order to solve our problem we will employ the method of Lagrange
multipliers. The Lagrangian is defined as
\[
  L = \norm{ \Delta\pmat{A}}_F^2 + \lambda^T
  \begin{pmatrix} \vec ( \Adj(\pmat{A}+\Delta \A) - \Fstar h)\\
    \mathcal{N}_h \vec(h) -1
  \end{pmatrix},
\]
where $\lambda$ is a vector of Lagrange multipliers.

A necessary first-order condition (KKT condition, e.g. \citep{Ber99})
for a tuple $z^\star = z^\star(\Delta \pmat{A}, \Fstar,h,\lambda)$ to be a 
regular
(attainable) minimizer is that the gradient of $L$ vanishes, that is,
\begin{equation}
  \label{eqn:first-order-nec}
  \nabla L(z^\star)=0.
\end{equation}
Let $J$ be the Jacobian matrix of the constraints defined as
\[
  J = \nabla_{\Delta \A,\Fstar,h}
  \begin{pmatrix}
    \vec(\Adj(\pmat{A}+\Delta \A)-\Fstar h)
  \end{pmatrix}.
\]
The second-order sufficiency condition for optimality at a local minimizer
$z^\star$ is that
\begin{equation}
  \label{eqn:second-order-suff}
  \ker(J(z^\star))^T \nabla^2_{xx} L(z^\star) \ker(J(z^\star)) \succ  0,
\end{equation}
that is, the Hessian with respect to $x=x(\Delta \A,\Fstar,h)$ is positive
definite over the kernel of the Jacobian of the constraints.  The vector
$x$ corresponds to the variables in the affine structure of
$\Delta \A$,$\Fstar$, and $h$.  If \eqref{eqn:first-order-nec} and
\eqref{eqn:second-order-suff} both hold, then $z^\star$ is necessarily a local
minimizer of \eqref{eqn:opt-prob}. Of course, it is also necessary that
$\ker(J(z^\star))^T \nabla^2_{xx} L(z^\star) \ker(J(z^\star)) \succeq 0$ at a
minimizer, which is the second-order necessary condition.
Our strategy for computing a local solution is to solve $\nabla L =0$ using a
Newton-like method.

\subsection{An Implementation with Local Quadratic 
Convergence}\label{ssec:fast-local-quadratic}
A problem with Newton-like methods is that when the Hessian is rank
deficient or ill-conditioned, then the Newton step becomes ill-defined or the
rate of convergence degrades. The proposed formulation of our problem can
encounter a rank deficient Hessian (this is due to over padding some
vectors with zero entries or redundant constraints). Despite this we are
still able to obtain a method with rapid local convergence under a very
weak normalization assumption.

In order to obtain rapid convergence we make use of the Levenberg-Marquart
(LM) algorithm.  If $H= \nabla^2 L$, then the LM iteration is defined as
repeatedly solving for $z^{(k+1)}=z^{(k)}+\Delta z^{(k)}$~by
\[
  (H^TH+\nu_k I) \Delta z^{(k)} = -H^T \nabla L(z^{(k)})
  \text{ where }
  z=
  \begin{pmatrix}
x \\
    \lambda
  \end{pmatrix} \in \RR^\ell,
\]
for some $\ell >0$ while using $\norm{\nabla L}_2$ as a merit function.  The
speed of convergence depends on the choice of $\nu_{k}>0$.  Note that since
LM is essentially a regularized Gauss-Newton method, when the Hessian is
rank deficient then we may converge to a stationary point of the merit
function. 
If convergence to a stationary point of the merit function is
detected, then the method of \cite{Wri05} can be used to replace LM in several 
instances.

\cite{YamFuk01} show that, under a local-error bound condition, a system of
non-linear equations $g(z)=0$ approximated by LM will converge
quadratically to a solution with a suitable initial guess.
Essentially, what this says is that to obtain rapid convergence it is
sufficient for regularity ($J$ having full rank) to hold or second-order
sufficiency, but it is not necessary to satisfy both.  Note 
that we assume Lagrange multipliers exist. However, unlike the case when $J$ has full rank, the 
multipliers need not be unique. The advantage of LM over other Newton-like methods
is that this method is globalized\footnote{Here ``globalized'' means that the method will converge to a stationary point of the merit function, not
  a local extremum of the problem.} in exchange for an extra matrix
multiplication, as $H^TH+\nu_k I$ is always positive definite, and hence
always a descent direction for the merit function.
We make the choice of $\nu_k \approx \norm{g(z)}_2$ based on the results of
\cite{FanYa05}.

\begin{definition}[Local Error Bound] \label{def:local-error-bound} Let
  $Z^\star$ be the set of all solutions to $g(z)=0$ and $X$ be a subset of
  $\RR^\ell$ such that $X\cap Z^\star\neq \emptyset$.  We say that
  $\norm{g(z)}$ provides a local error bound on $g(z) =0$ if there exists a
  positive constant $c$ such that $c \cdot \dist(z,Z^\star) \leq \norm{g(z)}$
  for all $z\in X$, where $\dist(\cdot)$ is the distance between a point
  and a set.
\end{definition}

In this section it is useful to consider $g(z) = \nabla L(z)$, as we need the 
local error bounds to estimate $\nabla L(z)=0$.

\begin{theorem}
  \label{thm:second-order-suff}
  If the second-order sufficiency condition \eqref{eqn:second-order-suff}
  \mbox{holds at} an attainable solution to \eqref{eqn:opt-prob}, then the
  local error-bound \mbox{property holds.}
\end{theorem}

\begin{proof}
 This result follows immediately from Section 3 of \cite{Wri05} and the 
references therein. 
\end{proof}

The bounds of Wright can be used to infer when quadratic convergence occurs for 
Newton-like methods. In this problem, perturbations to $x$ are important in 
understanding how the problem behaves locally. 
\begin{remark}

Let $z=z(x,\lambda)$
  where $x$ is a vector of variables and
  $\lambda$ is a vector of Lagrange multipliers,
  and define $g(z) = \nabla L(z)$.  First suppose that both the
  second-order sufficiency condition \eqref{eqn:second-order-suff} and
  first-order necessary condition \eqref{eqn:first-order-nec} hold at the
  point $z^\star$.  We can write the first-order expansion
  \[ g(z^\star + \Delta z) = H(z^\star)(\Delta z) + O(\norm{\Delta z}_2^2) 
\approx
    H(z^\star)(\Delta z),\] noting that $g(z^\star)=0$. 
It is
  useful to observe that
  \[ H(z^\star) =
    \begin{pmatrix}
      H_{xx}(z^\star) & J^T(z^\star) \\
      J(z^\star) & 0
    \end{pmatrix}.
  \]
If $\Delta x=0$ then the error-bound from~\cite{Hof52} (main theorem) applies
  and we have that there exists $c_{hof}>0$ such that
  $c_{hof}\norm{\Delta \lambda} \leq \norm{g(x,\lambda+\Delta
    \lambda)}$. If $\Delta x \neq 0$ then
  $\tallnorm{\begin{pmatrix} H_{xx}(z^\star) \\ J(z^\star)\end{pmatrix} \Delta x 
}
  \approx \norm{g(x+\Delta x,\lambda)}$ and \eqref{eqn:second-order-suff}
  implies that $H(z^\star) (\Delta z) = 0 \implies \Delta x=0$, so
  \[\sigma_{\min}\begin{pmatrix} H_{xx}(z^\star) \\ J(z^\star)\end{pmatrix}
    \norm{\Delta x} \lesssim \norm{g(x+\Delta x,\lambda)},\]
    so there exists $c_{\sigma_{\min}}>0$ when $\norm{\Delta x}$ is 
sufficiently small such that $c_{\sigma_{\min}}\norm{\Delta x} \leq 
\norm{g(x+\Delta x,\lambda)}$. Note that $c_{\sigma_{\min}}\approx 
\sigma_{\min}\begin{pmatrix} H_{xx}(z^\star) \\ J(z^\star)\end{pmatrix}.$

The first-order approximation implies that when $\norm{\Delta z}$ is 
sufficiently small that \[g(z^\star+\Delta z)\approx H(z^\star)(\Delta z)  =  
H(z^\star)\begin{pmatrix}
     \Delta x\\
     0
    \end{pmatrix} + H(z^\star)\begin{pmatrix}
     0\\
     \Delta \lambda
    \end{pmatrix} \approx g(x+\Delta x,\lambda) + g(x,\lambda+ \Delta 
\lambda).\]

The key idea is to separate the problem into the cases of
$\Delta x=0$ and $\Delta x \neq 0$, and then derive error bounds for each
case.  The important part of the discussion is that if one can estimate 
$c_{\sigma_{\min}}$ then one can often infer when quadratic convergence occurs.

The second-order 
sufficiency assumption is not necessary to derive
error bounds bounds. It is straightforward to show the local error bound
property holds if $J(z^\star)$ has full rank, as the Lagrange multipliers will 
be (locally) unique, hence the solution is (locally) unique. Alternatively, if 
$J$ had constant rank in a non-trivial open neighborhood around a solution, 
then a similar argument could be made about the local error-bound property. 
\end{remark}

\begin{theorem}\label{thm:second-order-suff-holds}
  The second-order sufficiency condition holds at minimal solutions with
  Lagrange multipliers of minimal norm if $h$ is of maximal degree and
  monic and the minimal structured perturbation $\norm{\Delta \A^\star}$ is
  sufficiently small.
\end{theorem}

\begin{proof}
  The Hessian of $L$ with respect to $x=x(\Delta A,\Fstar,h)$ is
  \[\nabla^2_{xx}L = H_{xx} =\begin{pmatrix}
      F+2I & & \\
      & & E \\
      & E^T &
    \end{pmatrix},
  \]
  where $F$ is a square matrix with zero diagonal whose entries are a
  multi-linear polynomial in $\lambda$ and $\Delta \A$ and $E^T$ is a
  matrix whose entries are homogeneous linear functions in $\lambda$.

  If $\Delta \A^\star=0$ then $\lambda^\star =0$. Hence both $E=0$ and $F=0$ and so,
  if $y \in \ker(H_{xx}) \cap \ker(J)$ then $y = \begin{pmatrix} 0 & y_2 & y_3
  \end{pmatrix}^T$.  
  The Jacobian of the constraints may be written (up to 
permutation) as 
\[ J = \begin{pmatrix}
      *& \mathcal{C}_{h} & \mathcal{C}_{\Fstar}  \\
      & & \mathcal{N}_{h}
    \end{pmatrix},
  \]
  where $*$ are blocks corresponding to differentiating with respect to
  variables in $\Delta \A$ and the blocks $\mathcal{C}_{\Fstar}$ and
  $\mathcal{C}_{h}$ consist of block convolution and convolution matrices that
  correspond to multiplication by $\Fstar$ and $h$, respectively. The block
  $\mathcal{N}_h$ contains a normalization vector to ensure that $h$ has the 
appropriate degree.
$Jy =0$ implies that there exists a vector of polynomials $v$ and a
  polynomial $u$ with the same degrees as $\Fstar$ and $h$ such that
  $\Fstar u + vh=0$ and $\mathcal{N}_h \vec(u) =0$.

  We have that $h$ is a factor of both $\Fstar u$ and $vh$. Since
  $\gcd(\Fstar,h)=1$ it must be that $h$ is a factor of $u$. It follows
  that $\deg (u)=\deg (h)$, so there exists some $\alpha \neq 0$ such that
  $\alpha u = h$.  Since $h$ is monic, we have that
  $\mathcal{N}_h\vec(h) =1$ but $\mathcal{N}_{h}\vec(u) =0$, which implies that
  $\alpha =0$, and so $u=0$. We have that $vh=0$ and
this implies $v=0$. Hence $\ker(J)\cap \ker(H_{xx}) = 0$ and
  second-order sufficiency holds when $\norm{\Delta \A ^*}=0$.

  If $\norm{\Delta \A ^*}$ is sufficiently small, then $\norm{F}$ will be
  sufficiently small so that $F+2I$ has full rank.  Accordingly, we have
  that
  \[ \ker \begin{pmatrix}
      F+2I & & \\
      & 0& E \\
      & E^T & 0
    \end{pmatrix} \subseteq \ker
    \begin{pmatrix}
      2I & & \\
      & 0&  \\
      & &0
    \end{pmatrix}. \qedhere\]
\end{proof}
We remark that the techniques in the proof are very similar to those of
\cite{ZenDay04} and \cite*{GieHarKal19} to show that a Jacobian matrix
appearing in approximate GCD computations of two (or more) polynomials has
full rank.  If we over-estimated the degrees of $\Fstar$ then $H_{xx}$
would have some columns and rows consisting of zero (the block-convolution
matrices would be padded with extra zero entries).

In the proof of Theorem~\ref{thm:second-order-suff-holds} we note that
\[ \nabla_{xx}^2L = \nabla^2_{xx}\norm{\Delta \A}_F^2 + \nabla_{x} \lambda^T
  J.\] The matrix $F=\nabla_{\Delta \A} \lambda ^T J_{\Adj}(\A+\Delta \A)$
will consist of coefficients of the $(n-3) \times (n-3)$ minors of
$\A+\Delta \A$ scaled by entries of $\lambda$.  Accordingly, $F$ will
generally not have $-2$ as an eigenvalue.

\begin{remark}
Thus far we have assumed that Lagrange multipliers exist at the current
solutions of interest, which are attainable solutions that have full
rank. Corollary~\ref{cor:adjoint-derivative-rank} and the proof of
Theorem~\ref{thm:second-order-suff-holds} imply that Lagrange
multipliers generally exist under these assumptions for several perturbation 
structures, since we need to solve
\[ \begin{pmatrix} 2\vec(\Delta \A)^T & 0
  \end{pmatrix}
  = -\lambda^T J,  \]
  of which $J$ generally has constant or full rank.
Of course if the solution was unattainable then
the GCD constraints would break down as there is a ``solution at infinity'' in 
a sense that $\norm{h}\to \infty$ as $\Delta \A \to \Delta \A^\star$. 
 \end{remark}

The implication of the local-error bound property holding is that one can
reasonably approximate when quadratic convergence occurs by estimating
$\sigma_{\min} \left( \left[ H_{xx} ~|~   J^T\right] \right)$
and $c_{hof}$.  In particular, these quantities act as a structured
condition number on the system. A structured backwards-error analysis of
existing techniques can be performed using these quantities.  Additionally,
it is somewhat generic that $F+2I$ has full rank, hence the local
error-bound will hold for most instances of the approximate SNF problem
with an attainable solution.
It is also important to note that we did not explicitly use the adjoint
matrix. Indeed the result remains valid if we replace the adjoint with
minors of prescribed dimension.  Likewise, if $\A$ is an ill-posed instance
of lower McCoy rank or approximate SNF without an attainable global
minimum, then optimizing over a reversal of each entry of
$\Adj(\A+\Delta \A)$ would yield a non-trivial answer and the same
stability properties would hold. Thus, poorly posed problems also remain
poorly posed if slightly perturbed.

\begin{corollary}
  The LM algorithm for solving $\nabla L=0$ has quadratic convergence under
  the assumptions of Theorem~\ref{thm:second-order-suff} and using
  $\nu_{k} = \norm{\nabla(L(z^k))}_2$.
\end{corollary}

\begin{proof}
  The quantity $\nabla L$ is a multivariate polynomial, hence it is locally
  Lipschitz. Second-order sufficiency holds, thus we have the local error
  bound property is satisfied.  The method converges rapidly with a
  suitable initial guess.
\end{proof}
Note that for several perturbation structures if the 
adjoint has generic degrees, then the Jacobian of the constraints will have 
full rank, and a 
standard Newton iteration is also well-defined, and will converge quadratically 
as well.

In the next section we discuss a technique that possibly forgoes rapid
local convergence, but has a polynomial per iteration cost to compute a low
McCoy rank approximation. 

\subsection{\mbox{Computational Challenges and Initial Guesses}}

The most glaring problem in deriving a fast iterative algorithm for the
approximate Smith form problem is that the
matrix $\Adj(\pmat{A}+\Delta \pmat{A})$ has exponentially many coefficients
as a multivariate polynomial in $\Delta \pmat{A}$.  This means computing
the adjoint matrix symbolically as an ansatz is not feasible.  In order to
solve \eqref{eqn:first-order-nec} we instead approximate the derivatives of
the coefficients of the adjoint numerically.

To compute an initial guess, we can use $\Delta \A _{init}=0$ and take
$\Fstar$ and $h$ to be a reasonable approximation to an approximate GCD of
$\Adj(\A)$, which will often be valid as per
Theorem~\ref{thm:second-order-suff}. To make sure the point is feasible, one 
can use a variant of Newton's method to project to a feasible point. 
Corollary~\ref{cor:adjoint-derivative-rank} implies that with a suitable 
initial guess, reasonable variants of Newton's method (such as LM) will 
converge quadratically to a feasible point, assuming one exists.

Another technique is to take two rows or columns of $\A$ and perturb them
so that the $2n$ entries have a non-trivial GCD. To find the best guess
with this technique, $O(n^2)$
approximate GCD
computations on $O(n)$ polynomials of degree $d$ need to be performed.
In the next section we will discuss more sophisticated techniques.

 \subsection{Attaining Unattainable Solutions}\label{sec:unattainable-soln}
 
 If a solution is unattainable then the degrees of all the entries of the
 adjoint matrix may change in an open neighborhood around a solution. If
 $\Delta \A^\star$ is an unattainable solution (of full rank) to
 \eqref{eqn:opt-prob} then $h(t)=t$ is clearly not a solution since $h(t)=t$ 
being a
 solution implies that such a solution would be attainable.  Let $d_{\Adj}$ be
 the generic degree of $\Adj(\A+\Delta \A)$, then $t$ is a factor of
 $\gcd(\rev_{d_{\Adj}} (\Adj(\A+\Delta \A^\star)))$. The reversed adjoint
 has no GCD at infinity by assumption, as such a GCD at infinity would be
 an attainable solution to the original problem.  Accordingly, we note
 that Theorem~\ref{thm:second-order-suff-holds} applies after some
 straightforward modifications, since
 \[
   \nabla \vec(\Adj(\A+\Delta \A)) \text{~~and~~} \nabla
   \vec(\rev_{d_{\Adj}}(\Adj(\A+\Delta\A)))
 \]
 are essentially (block) permutations of each other.

 Since $\rev_{d_{\Adj}}(\Adj(\A+\Delta\A)))$ achieves the generic degree,
 Lagrange multipliers should exist as we can apply
 Corollary~\ref{cor:adjoint-derivative-rank} on
 $\nabla \vec(\rev_{d_{\Adj}}(\Adj(\A+\Delta\A)))$ by permuting entries,
 and the underlying approximate GCD problem is well-posed. Thus the problem
 will also typically admit Lagrange multipliers.

 The essential ingredient in Theorem~\ref{thm:second-order-suff-holds} is
 the normalization of the underlying approximate GCD problem.
This means that ``backwards stable'' algorithms will compute the exact SNF
 of a nearby matrix polynomial that has no meaning in the context of
 computation. This generally occurs because the radius of uncertainty,
 usually proportional to unit rounding errors, contains infinitely many
 matrix polynomials with a non-trivial SNF.  The backwards stability is not
 meaningful in this context, because the instance of the problem is not
 continuous.  In such instances, computing the SNF is most likely the wrong 
\todo{how do I fix this?}
 problem to be considering. Instead, computing the spectral structure of 
eigenvalues at
 infinity is most likely the appropriate problem. However there exist
 instances where both problems could be simultaneously poorly conditioned.

 If the reversed problem has a radius of stability with respect to
 Theorem~\ref{thm:second-order-suff-holds}, then the original problem has a
 radius of instability, meaning that the iterates will converge to a point
 where $\norm{h}$ is excessively large.  In other words, if an instance of
 a problem is ill-posed, then it cannot be regularized --- the finite and
 infinite eigenvalues and their spectral structure is indistinguishable in
 floating point arithmetic --- in the context of the QZ decomposition,
 GUPTRI \citep{DemKaa93,DemKaa93b} or similar algorithms. There are some
 instances where attempting to compute the SNF numerically is not possible
 and should not be attempted.  In the context of an optimization problem,
 we can of course regularize the problem as we have just described. \cite{Van79} suggests that ill-posed problems should be formulated
 as an optimization problem as a means of regularization to overcome some
 of the numerical difficulties.

 \section{Lower McCoy Rank Approximation}
\label{sec:optII-linearization}

In this section we describe how to find a nearby matrix polynomial of lower
McCoy. 
Another way to formulate $\A$ having a non-trivial SNF is to solve the
minimization problem 
\begin{equation} \label{eqn:snf-via-linearization}
  \min
  \norm{\Delta \A}_F^2
  ~~~~~\mbox{subject to}~~~~~
  \parbox{16em}{$\bigl(\A (\omega) +\Delta \A(\omega)\bigr)B = 0$ ~and~ $B^*B = I_2$,\\[4pt]
    for some $\omega\in\CC$ and $B\in\CC^{n\times 2}$,}
\end{equation}
where $\Delta\A$ must have the appropriate structure.  Essentially
this finds the smallest perturbation of $\A$ with an eigenvalue that lowers the
rank by at least $2$. The auxiliary variables $\omega$ 
and $B$ are used to enforce this constraint. Here $B^*$ 
is the conjugate transpose of $B$, and
$B^*B=I_2$ ensures that the kernel vectors are linearly independent and 
do not tend towards zero.

The optimization is unstable if $\omega$ is reasonably large, since the
largest terms appearing are proportional to
$O((d+1)\norm{\A}_\infty |\omega|^d)$. To remedy this, if we assume that
a solution to the optimization problem \eqref{eqn:snf-via-linearization} 
exists and has full rank, then we may transform $\A+\Delta \A$ into a 
degree-one matrix polynomial (also known as a matrix pencil) with the same 
spectral 
properties, known as a {\emph{linearization}}. 
If there is no full-rank solution one can simply take a lower-rank
approximation \citep*{GieHarLab19} and extract a square matrix polynomial of 
full rank that may be linearized.  Alternatively, one may forgo the 
linearization and work directly with a problem that is more poorly conditioned. 
For the rest of this section we will assume, without loss of generality, that 
$\A$ and the 
solutions to the low McCoy rank problem have full rank. \todo{changed this}

We can encode the spectral structure and SNF of 
$\A$ as the following degree-one matrix polynomial (sometimes referred to as 
the \emph{companion linearization}~\citep{GolLanRod09})
of the form $\pmat{P}\in \RR[t]^{nd \times nd} $, defined as
\[ \pmat{P} = \begin{pmatrix}
    I \\
    & \ddots \\
    & & A_d
  \end{pmatrix} t -
  \begin{pmatrix}
    & I \\
    &   & \ddots \\
    -A_0 & -A_1 & \cdots & -A_{d-1}
  \end{pmatrix} .
\]
This particular linearization encodes the SNF of $\A$, as
$\SNF(\pmat{P}) = \diag(I,I,\ldots, I, \SNF(\A))$.  It follows that $\A$ has a non-trivial SNF if and only if $\pmat{P}$ has a non-trivial
SNF. If we preserve the affine structure of $\pmat{P}$ and only perturb
blocks corresponding to $\A$, then the reduction to a pencil will be
sufficient. Other linearizations are possible as well. The pencil is
generally better behaved numerically since the largest entry upon
evaluation at a $\omega\in\CC$ is proportional to
$O(d\norm{\A}_{\infty}|\omega|)$ rather than
$O(\norm{A}_{\infty}|\omega|^d)$, albeit with matrices that are $d$ times larger.

\subsection{Fast Low McCoy Rank via Optimization}

One way to approach the lower McCoy rank approximation problem is to study all
the minors (or sufficiently many) of a matrix polynomial. This method 
immediately generalizes from the previous section, however is not practical 
for computational purposes since the number of minors grows exponentially  in 
the dimension.  Instead, we can approach the problem by formulating it as an 
optimization problem, one that is remarkably similar to structured lower rank 
approximation of  scalar matrices.  This similarity facilitates computing an 
initial guess for the following optimization problem using the SVD. 

The lower McCoy rank approximation problem may be formulated as the
following \emph{real optimization problem}: to find the nearest matrix
polynomial to $\A\in\RR[t]^\nxn$ with McCoy rank $n-r$, find the
perturbation $\Delta\A\in\RR[t]^\nxn$ which minimizes
\begin{equation}
  \label{eqn:low-mccoy-rank-opt}
  \min \norm{\Delta \A}_F^2 \text{~~subject to~~}
  \begin{cases}
    \Re((\pmat{P} + \Delta \pmat{P})(\omega)B) =0, \\
    \Im((\pmat{P} + \Delta  \pmat{P})(\omega)B) = 0, \\
    \Re(B^*B) = I_r, \\
    \Im(B^*B) = 0
  \end{cases}
  ~~~~\text{for some $\omega\in\CC$ and $B\in\CC^{nd\times r}$.}
\end{equation}
Note that the perturbation $\Delta \A$ is {\emph{real valued}} in this
problem. The unitary constraint on $B$ ensures that $\rank(B) =r$ and each
column of $B$ remains away from zero. Accordingly, $\omega \in \CC$ will be
an eigenvalue of $(\P+\Delta \P)(\omega)$ since
$\rank((\P+\Delta \P)(\omega)) \leq nd-r$, and thus the McCoy rank of
$\A+\Delta \A$ is at-most $n-r$.

Real matrix polynomials can have complex eigenvalues and so complex numbers must 
necessarily appear in the constraints.  The constraints arising from the complex 
numbers may be divided   into real parts and imaginary parts, denoted as 
$\Re(\cdot)$ and $\Im(\cdot)$, respectively. By dividing the constraint into real 
and imaginary parts, we are able to solve an equivalent  optimization problem 
completely with real variables.   This ensures that $\Im(\Delta \A) =0$, that is, the 
perturbations are real. 
Since $\A+\Delta \A$ may have complex eigenvalues (but entries with real 
coefficients), we require that $\SNF(\A+\Delta \A)$  has entries from 
$\RR[t]$. Accordingly, we need to interpret the auxiliary variable $\omega$. 
The instance of $\Im(\omega) =0$ corresponds to $t-\omega$ as an invariant 
factor, while $\Im(\omega)\neq 0$ corresponds to the real irreducible quadratic 
$(t-\omega)(t- \overline{\omega})$.  Thus at a solution, we are able to recover 
a real invariant factor regardless if $\omega$ has a non-zero imaginary part. 

In order to approach the problem using the method of Lagrange multipliers
we define the Lagrangian as
\[ L = \norm{\Delta \A}_F^2 + \lambda^T\vec\begin{pmatrix}
    \Re((\pmat{P}+\Delta \pmat{P})(\omega)B)\\
    \Im((\pmat{P} + \Delta
    \pmat{P})(\omega)B) \\
    \Re(B^*B) - I_r \\
    \Im(B^*B) \end{pmatrix},
\]
and proceed to solve $\nabla L = 0$.  In our implementation we again make
use of the LM method, although given the relatively cheap gradient cost, a
first-order method will often be sufficient and faster.  The problem is
essentially tri-linear, and structurally similar to affinely structured low rank
approximation, of which Lagrange multipliers will exist for most instances.

It is important to note that an attainable solution to this problem is not
guaranteed, as it is possible for $\norm{\omega}\to \infty$ as
$\Delta \A \to \Delta \A^{\star}$.  Such an instance is an unattainable
solution in the context of Section~\ref{sec:unattainable-soln}. These
solutions behave like an infinite eigenvalue and can be handed by specifically 
considering the eigenvalue $t=0$  of the reversed matrix polynomial. 

\subsection{Computing an Initial Guess}
In order to compute an initial guess to \eqref{eqn:low-mccoy-rank-opt} we 
exploit the pseudo tri-linearity of the
problem. If two of $\Delta \A$, $\omega$ and $B$ are fixed then the problem is 
linear (or a linear surrogate can be solved) in the other variable. Despite the 
unitary constraint on $B$ 
being non-linear, it is not challenging to handle. Any full rank  $B$ is 
suitable for an initial guess, since we may orthonormalize $B$ to satisfy the 
constraint that $B^*B=I_r$.

First we approximate 
the determinant of $\A$ and consider initial
guesses where $\sigma_{n-r}(\A(\omega^{init}))$ is reasonably small. If 
$\sigma_{n-r}(\A(\omega^{init}))$ is reasonably small, then $\omega^{init}$ is 
(approximately) an eigenvalue of a nearby matrix polynomial of reduced McCoy 
rank. The zeros and local extrema of $\det(\A)$ are suitable candidates for 
computing an initial guess for $\omega$. The kernel $B^{init}$ can be 
approximated from the smallest $r$ singular vectors of $\A(\omega^{init})$. 
This ensures that $B^{init}$ is unitary and spans the kernel of a nearby rank 
deficient (scalar) matrix.

To compute an initial guess for $\Delta \A$ we can take $\Delta \A^{init} =0$, 
or solve a linear least squares problem where $B$ and $\omega$ are fixed. 
Alternatively, one may project to a feasible point by using a variant of 
Newton's method, using $\Delta \A^{init}=0$, $\omega^{init}$ and $B^{init}$ as 
an initial guess for the Newton iteration to solve $(\A+\Delta \A)(\omega)B 
=0$ and $B^*B = I_r$.  A feasible point computed by Newton's method tends not 
to perturb $\Delta \A$ very much, whereas the least squares approximation may 
perturb $\A$  by an unnecessarily large amount.

\subsection{About Global Optimization Methods}
The problems previously discussed are NP hard to solve exactly and to approximate with
coefficients from $\QQ$. This follows since affinely structured low rank
approximation \citep{BraYouDoyMor94,PolRoh93} is a special case. If we
consider a matrix polynomial of degree zero, then this is a scalar matrix
with an affine structure.  The approximate SNF will be a matrix of rank at
most $n-2$, and finding the nearest affinely structured singular matrix is
NP hard.

Despite the problem being intractable in the worst case, not all instances
are necessarily hard. The formulation \eqref{eqn:low-mccoy-rank-opt} is
multi-linear and polynomial, hence amenable to the sum of squares
hierarchy.  Lasserre's sum of squares hierarchy \citep{Las01} is a global
framework for polynomial optimization that asymptotically approximates a
lower bound.  Accordingly, if $\norm{\omega^{opt}}$ is bounded, then sum of
squares techniques should yield insight into the problem.

 \section{Implementation and Examples}
\label{sec:implementation}

We have implemented our algorithms and techniques in the
\texttt{Maple} computer algebra system\footnote{Sample code is at
\url{https://www.scg.uwaterloo.ca/software/GHL2018jsc-code-2018-11-28.tgz}.
}.
We use the variant of Levenberg-Marquardt discussed in
Section~\ref{sec:opt1-adjoint} in several instances to solve the first-order
necessary condition.  All computations are done using hardware precision
and measured in floating point operations, or FLOPs.  The input size of our
problem is measured in the dimension and degree of $\A$, which are $n$ and
$d$ respectively.  The cost of most quasi-Newton methods is roughly
proportional to inverting the Hessian matrix, which is $O(\ell^3),$ where
$\ell$ is the number of variables in the problem.

The method of Section~\ref{sec:opt1-adjoint} requires approximately
$O( (n^3d)^3) = O(n^{9}d^3)$ FLOPs per iteration in an asymptotically optimal 
implementation with cubic matrix inversion,  which is the cost of inverting the 
Hessian. Computing the 
Hessian costs roughly $\softO(n^4 d^2\times (n^2)^2 ) = \softO(n^8 d^2)$ FLOPs\footnote{For $\phi,\psi\colon \mathbb{R}\to\mathbb{R}$, $\phi=\softO(\psi)$ iff
$\phi=O(\psi (\log|\psi|)^c)$ for some absolute constant $c\geq 0$, i.e.,
we ignore log factors.}
using a blocking procedure, assuming the adjoint computation runs in  
$\softO(n^4d)$ FLOPs (which can be done via interpolation in a 
straightforward manner)There are $O(n^3d)$ Lagrange multipliers since the 
adjoint has degree at most $(n-1)d$.  Using reverse-mode automatic 
differentiation to compute $\nabla^2L$, this can be accomplished in 
$\softO(n^4d \times n^3d) = \softO(n^7 d^2)$ FLOPs.

The method of Section~\ref{sec:optII-linearization} has a Hessian matrix of
size $O(n^2d^2) \times O(n^2 d^2)$ in the case of a rank zero McCoy rank
approximation. Accordingly, the per iteration cost is roughly
$O(n^{6}d^{6})$ FLOPs. If the linearization is not performed, then the
per-iteration cost is $O(n^6d^3)$ FLOPs.  Given the lack of expensive
adjoint computation, a first-order method will typically require several
orders of magnitude fewer FLOPs per iteration (ignoring the initial setup
cost), with local linear convergence.

\begin{example}[Nearest Interesting SNF]
  Consider the matrix polynomial $\A$ with a trivial SNF
  \[
    \begin{pmatrix} {t}^{2}+ .1 t+1&0& .3 t- .1&0\\
      0& .9 {t}^{2}+ .2 t+ 1.3&0& .1
      \\   .2 t&0&{t}^{2}+ 1.32+ .03 {t}^{3}&0\\
      0& .1 {t}^{2}+ 1.2&0& .89 {t}^{2}+ .89
    \end {pmatrix}
  \]
  of the form $\diag(1,\ldots,1,\det(\A))$.
  
  If we prescribe the perturbations to leave zero coefficients unchanged,
  then using the methods of Section~\ref{sec:opt1-adjoint} and
  Section~\ref{sec:optII-linearization} results in  a local minimizer
  $\A+\Delta \A_{opt}$ given by
  \[\scalemath{.9}{ \begin{pmatrix} 1.0619 {t}^{2}+ .018349 t+ .94098& 0&
        .27477 t- .077901& 0\\ 0& .90268 {t}^{2}+ .22581 t+ 1.2955& 0&
        .058333 \\ .13670 t& 0& .027758 {t}^{3}+ .97840 {t}^{2}+ 1.3422&
        0\\ 0& .10285 {t}^{2}+ 1.1977& 0& .84057 {t}^{2}+
        .93694 \end{pmatrix} },\]
  with $\norm{\Delta \A_{opt}}\approx .164813183138322$.  The SNF of
  $\A+\Delta\A_{opt}$ is approximately
  \[\diag(1, 1,s_1,s_1( t^5 + 35.388 t^4 + 6.4540 t^3 + 99.542 t^2 + 5.6777 t
    + 70.015)),\] where
  $s_1 \approx t^2 + 0.0632934647739423t + 0.960572576466186$.
  The factor $s_1$
  corresponds to
  $\omega_{opt}\approx -0.0316467323869714 - 0.979576980535687 i$.

  The method discussed in Section~\ref{sec:opt1-adjoint} converges to
  approximately $14$ decimal points of accuracy\footnote{$\nabla L=0$ is
    solved to $14$ digits of accuracy; the extracted quantities are accurate
    to approximately the same amount.} after $69$ iterations and the method
  of Section~\ref{sec:optII-linearization} converges to the same precision
  after approximately $34$ iterations.  The initial guess used in both instances
  was $ \Delta\A_{init}=0$. The initial guesses of $\Fstar$ and $h$ were
  computed by an approximate GCD routine. For the initial guess of $\omega$
  we chose a root or local extrema of $\det(\A)$ that minimized the
  second-smallest singular value of $\A(\omega)$, one of which is
  $\omega_{init}\approx -.12793 - 1.0223 i$.
\end{example}

\begin{example}[Lowest McCoy Rank Approximation]
  Let $\A$ be as in the previous example and consider the $0$-McCoy rank
  approximation problem with the same prescribed perturbation structure.

  In this case we compute a local minimizer $\A+\Delta \A_{opt}$ given by
  \[
    \scalemath{1}{
      \begin{pmatrix}.80863 {t}^{2}+ 1.1362& 0& 0& 0\\
        0& .91673 {t}^{2}+1.2881& 0& 0\\   0& 0& .95980 {t}^{2}+1.3486 & 0\\
        0& .60052 {t}^{2}+ .84378& 0& .71968 {t}^{2}+ 1.0112
      \end{pmatrix}},
  \]
  with $\norm{\Delta \A_{opt}}\approx .824645447014665$ after $34$ iterations
  to $14$ decimal points of accuracy.  We compute
  $\omega_{opt}\approx - 1.18536618732372 i$ which corresponds to the
  single invariant factor $s_1 \approx t^2 + 1.4051$. The SNF of
  $\A+\Delta \A_{opt}$ is of the form $(s_1,s_1,s_1,s_1)$.
  
\end{example}

 \section{Conclusion and Topics for Future Research}

In this paper we have shown that the problem of computing a nearby matrix polynomial with a 
non-trivial spectral structure can be solved by (mostly local) optimization techniques. 
Regularity conditions were shown to hold for most instances of the 
problems in question, ensuring that Lagrange multipliers exist under mild 
assumptions about the solutions. When Lagrange multipliers do not exist, 
alternative formulations that admit Lagrange  multipliers have been proposed.  
Several of these algorithms are shown to be theoretically  robust with a 
suitable initial guess. In general, reasonable quasi-Newton methods will have 
rapid local convergence under normalization assumptions for all the problems 
considered. 

There are a number of problems that remain open for future work. In particular
in the case of nearby nontrivial Smith forms there is the question of obtaining
such forms via polynomial row and column operations, that is, finding the unimodular matrix
multipliers that will produce our nearest Smith form. Preliminary work on this topic,
including the formulation as an optimization problem and the proving of the existence 
of Lagrange multipliers for the optimization can be found in the thesis of \citet{Har19}.
In some cases it may be practical to prescribe the degree structure, 
also called the {\emph{structural supports}}, of the eigenvalues or the 
invariant factors of a nearby matrix polynomial. In this case, rather than look for 
a closest non-trivial SNF one would be interested in a closest SNF having a particular degree structure.
As before this can be formulated as an optimization problem with early results available in \citep{Har19}.

\bibliographystyle{elsarticle-harv.bst}

\newcommand{\Gathen}{\relax}

\end{document}